\documentclass[a4paper, 12pt]{article}

\usepackage{amsmath, amssymb, amsthm}
\usepackage[english]{babel}
\usepackage[T1]{fontenc}

\newcommand{\be}{\begin{equation}}
\newcommand{\ee}{\end{equation}}
\newcommand{\bd}{\begin{displaymath}}
\newcommand{\ed}{\end{displaymath}}
\newcommand{\ba}{\begin{eqnarray}}
\newcommand{\ea}{\end{eqnarray}}

\def\R{{I \!\! R}}

\def\a{\alpha}

\def\b{\beta}
\def\v12{(v-w)}

\def\({\left(}
\def\){\right)}

\def\bgr#1\egr{{\allowdisplaybreaks\begin{gather}#1\end{gather}}}
\def\bma#1\ema{{\allowdisplaybreaks\begin{align}#1\end{align}}}

\def\oplem#1{\begin{lemma}\, {\rm #1}\, \it }
\def\cllem{\end{lemma}\rm \par }
\def\opthm#1{\begin{theorem}\, {\rm #1}\, \it }
\def\clthm{\end{theorem}\rm \par }

\def\R{\mathbb{R}}

\newcommand{\fer}[1]{(\ref{#1})}
\newcommand{\bq}{\begin{equation}}
\newcommand{\eq}{\end{equation}}
\def\bqa{\begin{eqnarray}}
\def\eqa{\end{eqnarray}}
\def\bd{\begin{displaymath}}
\def\ed{\end{displaymath}}

\newtheorem{thm}{Theorem}

\newtheorem{lem}[thm]{Lemma}

\theoremstyle{remark}

\theoremstyle{definition}

\begin{document}

\title{Heat equation and the sharp Young's inequality}

 \author{Giuseppe Toscani \thanks{Department of Mathematics, University of Pavia, via Ferrata 1, 27100 Pavia, Italy.
\texttt{giuseppe.toscani@unipv.it} }}

\maketitle

\begin{center}\small
\parbox{0.85\textwidth}{
\textbf{Abstract.} We  show that the sharp Young's inequality for
convolutions first obtained by Bechner \cite{Bec} and Brascamp-Lieb
\cite{BL} can be derived from the monotone in time evolution of a Lyapunov
functional of the convolution of two solutions to the heat equation, with
different diffusion coefficients, first introduced by Bennett and Bez in
\cite{BB}. Our proof is based on a suitable adaptation of an old idea of
Stam \cite{Sta} and Blachman \cite{Bla}, used to obtain  Shannon's entropy
power inequality.

}
\end{center}

\medskip

\section{Introduction}

The goal of this note is to present a new proof of the Young inequality in
the sharp form obtained by Bechner \cite{Bec},
 \be\label{young}
\| f*g\|_r \le (A_pA_qA_{r'})^n\| f\|_p\| g\|_q .
 \ee
In \fer{young}  $f\in L^p(\R^n)$,  $g\in L^q(\R^n)$, $1 < p,q,r <
\infty$ and $1/p + 1/q = 1 +1/r$. Moreover, the constant $A_m$ which
defines the sharp constant is given by
 \be\label{c+}
 A_m = \left(\frac{m^{1/m}}{m'^{1/m'}} \right)^{1/2}
 \ee
where primes always denote dual exponents, $1/m + 1/m' = 1$.

The best constants in Young's inequality were found by Beckner \cite{Bec},
using tensorisation arguments and rearrangements of functions. In
\cite{BL}, Brascamp and Lieb derived them from a more general inequality,
which is nowadays known as the Brascamp-Lieb inequality. The expression of
the best constant, in the case in which both $f$ and $g$ are probability
density functions, is obtained by noticing that inequality \fer{young} is
saturated by Gaussian densities. This principle has been largely utilized
by Lieb in a more recent paper \cite{Li}. Among many other results, this
paper contains a new proof of the Brascamp-Lieb inequality. In \cite{BL},
Brascamp and Lieb noticed that the sharp form of Young inequality also
holds in the so-called reverse case
  \be\label{young-r}
\| f*g\|_r \ge (A_pA_qA_{r'})^n\| f\|_p\| g\|_q ,
 \ee
where now $0 < p,q,r < 1$ while, as in Young inequality \fer{young}, $1/p
+ 1/q = 1 +1/r$. In this case, however, the dual exponents $p',q',r'$ are
negative, and
 \be\label{c-}
 A_m = \left(\frac{m^{1/m}}{|m'|^{1/|m'|}} \right)^{1/2}.
 \ee
The proof of this sharp reverse Young inequality was subsequently
simplified by Barthe \cite{Ba}. While the original proof in \cite{BL} was
rather complicated, and used tensorisation, Schwarz symmetrization,
Brunn-Minkowski and some not so intuitive phenomenon for the measure in
high dimension, the new proof in \cite{Ba} was based on relatively more
elementary arguments and gave a unified treatment of both cases, the Young
inequality \fer{young} and its reverse form \fer{young-r}. As a matter of
fact, the proof of the main result in \cite{Ba} relies on a
parametrization of functions which was used in \cite{HM} and was suggested
by Brunn's proof of the Brunn-Minkowski inequality.

In a recent paper,  Young's inequality  has been seen in a different light
by Bennett and Bez \cite{BB}. In their paper, Young's inequality is
derived by looking at the monotonicity of a suitable functional of the
solution to the heat equation. Even if not explicitly mentioned in the
paper, this idea connects Young's inequality in sharp form with other
inequalities, for which the proof exactly moved along the same idea.

The connections of the sharp form of Young inequality with other
inequalities has been enlightened by Lieb in \cite{Lieb}. He proved in
fact that, by letting $p,q,r \to 1$ in \fer{young}, the sharp form of
Young's inequality reduces to another well-known inequality in information
theory, known as Shannon's entropy power inequality \cite{Sha}.

In its original version, Shannon's entropy power inequality  gives a
lower bound on Shannon's entropy functional of the sum of
independent random variables $X, Y$ with densities
 \be\label{entr}
\exp\left(2H(X+Y)\right) \ge \exp\left(2H(X)\right)+ \exp \left(
2H(Y)\right),
 \ee
with equality if $X$ and $Y$ are Gaussian random variables. Shannon's
entropy functional of the probability density function $f(x)$ of $X$ is
defined by
 \be\label{h-f}
H(X) = H(f) = - \int_\R f(x) \log f(x)\, dx.
 \ee
Note that Shannon's entropy functional coincides to Boltzmann's
$H$-functional \cite{Cer} up to a change of sign.  The entropy-power
 \[
 N(X) = N(f) = \exp\left(2H(X)\right)
 \]
(variance of a Gaussian random variable with the same Shannon's
entropy functional) is maximum and equal to the variance when the
random variable is Gaussian, and thus, the essence of \fer{entr} is
that the sum of independent random variables tends to be \emph{more
Gaussian} than one or both of the individual components.

The first rigorous proof of inequality \fer{entr} was given by Stam
\cite{Sta} (see also Blachman \cite{Bla} for the generalization
 to $n$-dimensional random vectors), and was based on an
identity which couples Fisher's information with Shannon's entropy
functional \cite{CTh}.

The original proofs of Blachman and Stam make a substantial use of the
solution to the heat equation
 \be\label{heat}
\frac{\partial f (x,t)}{\partial t} =  \Delta f(x,t),
 \ee
that is, for $t \ge 0$, of the function $f(x,t)= f*M_{2t}(x)$, where
$M_t(x)$ denotes the Gaussian density in $\R^n$ of variance $t$
 \be\label{max}
M_t(x) = \frac 1{(2\pi t)^{n/2}}\exp\left(\frac{|v|^2}{2t}\right).
 \ee
Other variations of the entropy--power inequality are present in the
literature. Costa's strengthened entropy--power inequality
\cite{Cos}, in which one of the variables is Gaussian, and a
generalized inequality for linear transforms of a random vector due
to Zamir and Feder \cite{ZF}.

Also, other properties of Shannon's entropy-power $N(f)$ have been
investigated so far. In particular, the \emph{concavity of entropy
power} theorem, which asserts that
 \be\label{conc}
\frac{d^2}{dt^2}\left(N(f*M_{2t})\right) \le 0
 \ee
Inequality \fer{conc} is due to Costa \cite{Cos}.  More recently, a
short and simple proof of \fer{conc} has been obtained by Villani
\cite{Vil}, using an old idea by McKean \cite{McK}.

Summarizing, the proof of Stam is based on the following argument.
Let $f(x,t)$ and $g(x,t)$ be two solutions of the heat equation
\fer{heat} corresponding to the initial data $f(x)$ (respectively
$g(x)$). If the entropies of the initial data are finite, one
considers the evolution in time of the  functional $\Theta_{f,g}
(t)$ defined by
 \be\label{stam}
\Theta_{f,g}(t) = \frac{\exp\{2H(f(t))\} +
\exp\{2H(g(t))\}}{\exp\{2H(f(t)*g(t))\}}.
 \ee
Evaluating the time derivative of $\Theta_{f,g}(t)$, and using a key
inequality for Fisher information on convolutions, shows that
$\Theta_{f,g}(t)$ is increasing in time, and converges towards the
constant value $\Theta_{f,g}(+\infty) = 1$, thus proving inequality
\fer{entr}. Note that this method of proof also determines the cases
of equality in \fer{entr}. It is interesting to remark that the
evaluation of the limit of $\Theta_{f,g}(t)$, as $t\to \infty$, is
made easy in reason of a scaling property. Indeed, the (Lyapunov)
functional $\Theta(f,g)$ is invariant with respect to the scaling
 \be\label{scal}
 f(x) \to f_a(x) = \frac 1{a^n} f\left( \frac x{a^n}\right), \quad a >0,
 \ee
which preserves the total mass of the function $f$. The importance
of this property will be clarified later on.

The proof by Stam is a \emph{physical} proof, in the spirit of
Boltzmann $H$-theorem \cite{Cer} in kinetic theory of rarefied
gases, where convergence towards the Maxwellian equilibrium is shown
in consequence of the monotonicity in time of the logarithmic
entropy \fer{h-f}.

In the rest of the paper, inspired by the Stam's approach to the proof of
Shannon's entropy power inequality, we will present a \emph{physical}
proof of both direct and reverse Young's inequalities, which is based on
the two ingredients specified above: a suitable use of solutions to the
heat equations, coupled with the scaling invariance property \fer{scal}.
Our proof is alternative to the proof of \cite{BB}, and relies on a result
which generalizes Stam proof of subadditivity of Fisher information.

In the following Section \ref{hol} we will describe how the method works
by proving H\"older inequality. Even if this is a well-known result
\cite{Car, BB, B1}, it will give indication on the underlying methodology.
Next, Section \ref{you} will contain our main result.

To make computations as simple as possible, we will present all
proofs in one dimension. Without loss of generality, in fact, one
can easily argue that identical proofs hold in dimension $n$, with
$n>1$. Also, if not strictly necessary, we will restrict ourselves
to consider as functions probability density functions. Since our
computations involve solutions to the heat equation, all details
involving regularity and the various integrations by parts can be
assumed to hold true.

The basic idea used here is that many inequalities can be viewed as the
tendency of various Lyapunov functionals of the solution to the heat
equation to reach their maximum value as time tends to infinity. The
discovery of a Lyapunov functional which allows to prove Young inequality
is only one of the possible application of this idea \cite{B1, BB, B2}. In
a forthcoming paper \cite{Tos}, we are going to develop this strategy by
revisiting various well-known inequalities in terms of the monotonicity of
suitable Lyapunov functionals.

\section{Heat equation and H\"older inequality}\label{hol}

We begin by showing that H\"older's inequality can be viewed as a
consequence of the time monotonicity of a suitable Lyapunov functional of
the solution to the heat equation \cite{Car,BB}.

H\"older's inequality for integrals states that, if $p,q
>1$ are such that $1/p + 1/q = 1$
 \be\label{holder}
 \int_\R |f(x)g(x)| \, dx \le \left( \int_\R |f(x)|^p\, dx \right)
 ^{1/p}\left( \int_\R |g(x)|^q\, dx \right)
 ^{1/q}.
 \ee
Moreover, there is equality in \fer{holder} if and only $f$ and $g$ are
such that there exist positive real numbers $a$ and $b$ such that $af^p(x)
= bg^q(x)$ almost everywhere. H\"older's inequality can be proven in many
ways, for example resorting to Young's inequality for constants, which
states that, if $1/p + 1/q = 1$
 \be\label{yo}
cd \le \frac{c^p}p + \frac{d^q}q,
 \ee
for all nonnegative $c$ and $d$, where equality is achieved if and only if
$c^p = d^q$.

Without loss of generality, one can assume that the functions $f,g$
in \fer{holder} are nonnegative. A different way to achieve
inequality \fer{holder} is contained into the following

\begin{thm} Let $\Phi_{u,v}(t)$ be the  functional
 \be\label{f-hol}
\Phi_{u,v}(t) = \int_\R u(x,t)^{1/p}v(x,t)^{1/q} \, dx,
 \ee
where $1/p + 1/q = 1$, and  $u(x,t)$ and $v(x,t)$, $t > 0$, are
solutions to the heat equation corresponding to the initial values
$u(x)\in L^1(\R)$ (respectively $v(x)\in L^1(\R)$). Then
$\phi_{u,v}(t)$ is increasing in time from
 \[
\Phi_{u,v}(t=0) = \int_\R u(x)^{1/p}v(x)^{1/q} \, dx,
 \]
to
 \[
\lim_{t \to \infty} \Phi_{u,v}(t) = \int_\R u(x) \, dx \, \int_\R
v(x) \, dx.
 \]
\end{thm}

\begin{proof}

We outline that the functional $\Phi_{u,v}(t)$ is invariant with
respect to the scaling \fer{scal}. Moreover, the condition $u(x),
v(x)\in L^1(\R)$ is enough to ensure that $\Phi_{u,v}(t)\in L^1(\R)$
at any time $t\ge 0$. Indeed, inequality \fer{yo} implies
 \[
u(x,t)^{1/p}v(x,t)^{1/q} \le \frac 1p u(x,t) + \frac 1q v(x,t),
 \]
where, since $u(x,t)$ and $v(x,t)$ are solution to the heat
equation,
 \[
\int_\R u(x,t)\, dx = \int_\R u(x)\, dx, \quad \int_\R v(x,t)\, dx =
\int_\R v(x)\, dx.
 \]
Let us first proceed in a formal way. However, by resorting to the
smoothness properties of the solution to the heat equation, all
mathematical details can be rigorously justified.

Let us evaluate the time derivative of $\Phi(t)$. It holds
 \[
{\Phi'}_{u,v}(t) = \int_\R \left[ (u(x,t)^{1/p})_t v(x,t)^{1/q} +
u(x,t)^{1/p}(v(x,t)^{1/q})_t \right]\, dx =
 \]
 \[
\int_\R \left[ \frac 1p u^{1/p-1}  v^{1/q} u_{xx} + \frac 1q
u^{1/p}v^{1/q-1}v_{xx} \right]\, dx = \int_\R \left[ \frac 1p
u^{-1/q} v^{1/q} u_{xx} + \frac 1q u^{1/p}v^{-1/p}v_{xx} \right]\,
dx.
 \]
Integrating by parts we end up with
 \[
{\Phi'}_{u,v}(t) = \frac 1{pq} \int_\R u^{1/p} v^{1/q} \left[
\left( \frac{u_{x}}u \right)^2 - 2 \frac{u_{x}}u  \frac{v_{x}}v  +
\left( \frac{v_{x}}v \right)^2  \right]\, dx =
 \]
 \be\label{sale}
 \frac 1{pq} \int_\R u^{1/p}(x,t)
v^{1/q}(x,t) \left( \frac{u_{x}(x,t)}{u(x,t)} - \frac{v_{x}(x,t)}{v(x,t)}
\right)^2\, dx \ge 0.
 \ee
Hence the functional $\Phi_{u,v}(t)$ is increasing in time. Note
that the time derivative of the functional is equal to zero if and
only if, for every $t >0$
 \[
\frac{u_{x}(x,t)}{u(x,t)} - \frac{v_{x}(x,t)}{v(x,t)} = 0
 \]
for all points $x \in \R$. This condition can be rewritten as
 \[
\frac d{dx} \log \frac{u(x,t)}{v(v,t)} = 0.
 \]
Consequently $\Phi'(t) = 0$ if and only if
 \be\label{equ}
u(x,t) = c\, v(x,t)
 \ee
for some positive constant $c$. Thus, unless condition \fer{equ} is
verified almost everywhere at time $t = 0$, the functional $\Phi(t)$
is monotone increasing, and it will reach its eventual maximum value
as time $t\to \infty$. The computation of the limit value uses in a
substantial way the scaling invariance of $\Phi$. In fact, at each
time $t>0$, the value of $\Phi_{u,v}(t)$ does not change if we scale
$u(x,t)$ and $v(x,t)$ according to
\begin{equation}\label{FP}
 \begin{split}
 &u(x,t) \to U(x,t) = \sqrt{1+2t}\, u(x\, \sqrt{1+2t}, t)  \\
& v(x,t) \to V(x,t) = \sqrt{1+2t}\, v(x\, \sqrt{1+2t},t).\\
\end{split}
\end{equation}
On the other hand, it is well-known that \cite{CT}
 \be\label{limi}
\lim_{t\to \infty} U(x,t) = M_1(x) \, \int_\R u(x) \, dx  \quad \lim_{t\to
\infty} V(x,t) = M_1(x)\, \int_\R v(x) \, dx ,
 \ee
where, according to \fer{max} $M_1(x)$ is the Gaussian density in $\R$ of
variance equal to $1$. Therefore, passing to the limit one obtains
 \[
\lim_{t \to \infty }\Phi_{u,v}(t) = \left(\int_\R u(x)\,
dx\right)^{1/p}\left( \int_\R v(x)\, dx \right)^{1/q} \int_\R
M_1(x)^{1/p} M_1(x)^{1/q} \, dx =
 \]
 \[
\left(\int_\R u(x)\, dx\right)^{1/p}\left( \int_\R v(x)\,
dx\right)^{1/q} \int_\R M_1(x) \, dx = \left(\int_\R u(x)\,
dx\right)^{1/p}\left( \int_\R v(x)\, dx\right)^{1/q}.
 \]
Since
 \[
\lim_{t \to 0^+ }\Phi_{u,v}(t) = \int_\R u(x)^{1/p}v(x)^{1/q} \, dx,
 \]
the monotonicity of the functional $\Phi(t)$ implies the inequality
 \be\label{ho1}
\int_\R u(x)^{1/p}v(x)^{1/q} \, dx \le \left(\int_\R u(x)\,
dx\right)^{1/p}\left( \int_\R v(x)\, dx \right)^{1/q},
 \ee
with equality if and only if \fer{equ} is verified at time $t =0$,
that is
 \be\label{eq1}
u(x) = c v(x),
 \ee
for some positive constant $c$. Setting $ f= u^{1/p}$ and $g =
v^{1/q}$ proves both H\"older inequality \fer{holder} and the
equality cases.
\end{proof}

Despite its apparent complexity, this way of proof is based on a
solid physical argument, namely the monotonicity in time of a
Lyapunov functional of the solution to the heat equation. This gives
a clear indication that many inequalities reflect the physical
principle of the tendency of a system to move towards the state of
maximum entropy. In the next Section we will see how this idea
applies to prove Young's inequality.

\section{Young's inequality and Lyapunov functionals}\label{you}

The proof of the sharp Young's inequality follows along the same lines of
the proof of H\"older's inequality we presented in Section \ref{hol}. In
this case the key functional to study is the one considered by Bennett and
Bez \cite{BB}
 \be\label{key1}
\Psi_{u,v}(t) = \left( \int_\R \left(
u(x,t)^{1/p}*v(x,t)^{1/q}\right)^r \, dx \right)^{1/r},
 \ee
where, as in Young's inequality, $1/p + 1/q = 1 +1/r$. With respect
to the notations of the previous Section, there is a substantial
difference in the meaning of the functions $u(x,t)$ and $v(x,t)$.
Here $u(x,t)$ and $v(x,t)$ are still solutions of the heat equation
corresponding to the initial data $u(x)$ (respectively $v(x)$).
However, these solutions correspond to two different heat equations,
with different coefficients of diffusions, say $\alpha$ and $\beta$.
In  other words, $u(x,t)$ solves the diffusion equation
 \be\label{hu}
u_t = \a u_{xx},
 \ee
while $v(x,t)$ solves
 \be\label{hv}
v_t = \b v_{xx}.
 \ee
Hence $u(x,t)$ and $v(x,t)$ diffuse at different velocities. It is a
simple exercise to verify that, in view of the relationship between
$p,q$ and $r$, the functional $\Psi_{u,v}(t)$ is invariant with
respect to the mass preserving scaling \fer{scal}.

\begin{thm}\label{th-young}  Let $\Psi_{u,v}(t)$ be the  functional
\fer{key1},
 where $1/p + 1/q = 1+1/r$, and  $u(x,t)$ and $v(x,t)$, $t > 0$, are
solutions to the heat equation corresponding to the initial values
$u(x)\in L^1(\R)$ (respectively $v(x)\in L^1(\R)$). Then, if $p,q,r >1$,
and the diffusion coefficients in \fer{hu} and \fer{hv} are given by $\a =
q'/p$ (respectively $\b = p'/q$), or by a multiple of them,
$\Psi_{u,v}(t)$ is increasing in time from
 \[
\Psi_{u,v}(t=0) = \left( \int_\R \left(u(x)^{1/p}* v(x)^{1/q}\right)^r \,
dx \right)^{1/r},
 \]
to the limit value
 \be\label{lim11}
\lim_{t \to \infty} \Psi_{u,v}(t) = \left( A_pA_q A_{r'}\right)^{1/2}
\left( \int_\R u(x)\, dx \right)^{1/p} \left( \int_\R v(x)\, dx
\right)^{1/q}.
 \ee
If on the contrary   $0 <p,q,r <1$, and the diffusion coefficients in
\fer{hu} and \fer{hv} are given by $\a = |q'|/p$ (respectively $\b =
|p'|/q$), or by a multiple of them, $\Psi_{u,v}(t)$ is decreasing in time
from
 \[
\Psi_{u,v}(t=0) = \left( \int_\R \left(u(x)^{1/p}* v(x)^{1/q}\right)^r \,
dx \right)^{1/r},
 \]
to the limit value \fer{lim11}, where now $A_m$ is given by \fer{c-}. In
both cases $\Psi'_{u,v}(t) = 0$ if and only if $u(x,t)$ and $v(x,t)$ are
Gaussian.
\end{thm}

\begin{proof}

Let us consider first the case in which $p,q,r >1$. Without loss of
generality, we will assume that both the initial data $u(x)$ and
$v(x)$ are probability density functions. This is sufficient to show
that, for any time $ t >0$, the functional $\Psi_{u,v}(t)$ is
bounded. Indeed we can write
 \[
\int_\R u(x-y)^{1/p}v(y)^{1/q} \, dy =
 \]
 \[
\int_{\{ u(x-y) \le v(y)\}} u(x-y)^{1/p} v(y)^{1/q} \, dy + \int_{\{u(y)
> v(x-y)\}} u(y)^{1/p}v(x-y)^{1/q} \, dy.
 \]
Now, since $r>1$, and $v(x,t)$ has mass equal to $1$, Jensen's
inequality implies
 \[
\left(\int_{\{ u(x-y) \le v(y)\}} u(x-y)^{1/p} v(y)^{1/q} \, dy \right)^r
=
 \]
 \[
\left(\int_{\{ u(x-y) \le v(y)\}} u(x-y)^{1/p} v(y)^{1/q-1} v(y) \, dy
\right)^r \le
 \]
 \[
\int_{\{ u(x-y) \le v(y)\}} u(x-y)^{r/p} v(y)^{r/q-r} v(y)  \, dy  =
 \]
 \[
\int_{\{ u(x-y) \le v(y)\}} u(x-y)^{r/p} v(y)^{r/q-r +1}  \, dy .
 \]
Note that
 \[
\frac rp + \frac rq -r + 1 = 2, \quad \frac rp > 1.
 \]
Therefore, on the set $ \{ u(x-y) \le v(y)\}$, since the exponent of
$v(y)$ is smaller than $1$,
 \be\label{bo1}
u(x-y)^{r/p} v(y)^{r/q-r +1} \le u(x-y)v(y).
 \ee
Inequality \fer{bo1} follows simply dividing by $u^2$. Therefore
  \[
\left(\int_{\{ u(x-y) \le v(y)\}} u(x-y)^{1/p} v(y)^{1/q} \, dy \right)^r
\le \int_{\{ u(x-y) \le v(y)\}} u(x-y) v(y) \, dy  \le 1 .
 \]
Identical computations show that
   \[
\left(\int_{\{ u(y) > v(x-y)\}} u(y)^{1/p} v(x-y)^{1/q} \, dy \right)^r
\le \int_{\{ u(x-y) \le v(y)\}} u(y) v(x-y) \, dy  \le 1.
 \]
 Therefore
 \[
\int_\R \left(u(x)^{1/p}* v(x)^{1/q}\right)^r  \le 2 c_r,
 \]
where $c_r$ is the positive constant in the inequality
 \[
(a+b)^r \le c_r(a^r + b^r).
 \]
We proceed now to compute the time derivative of the functional
$\Psi_{u,v}(t)$.  To shorten, let us denote
 \be\label{def1}
 h(x,t) = u(x,t)^{1/p}*v(x,t)^{1/q}.
 \ee
Since $u(x,t)$ and $v(x,t)$ are solutions to the heat equation
 \[
\frac{\partial}{\partial t}h(x,t)= \frac{\partial}{\partial t} \int_\R
u(x-y,t)^{1/p}v(y,t)^{1/q} \, dy =
 \]
 \[
\frac 1p \int_\R u(x-y,t)^{1/p-1}u_t(x-y,t )v(y,t)^{1/q} \, dy + \frac 1q
\int_\R u(x-y,t)^{1/p}v(y,t)^{1/q-1}v_t(y,t ) \, dy =
 \]
 \[
\frac{\a}p \int_\R u(x-y,t)^{1/p-1}u_{xx}(x-y,t )v(y,t)^{1/q} \, dy +
\frac{\b}q \int_\R u(x-y,t)^{1/p}v(y,t)^{1/q-1}v_{yy}(y,t ) \, dy .
 \]
On the other hand, we have
 \[
\frac{\partial^2}{\partial x^2}h(x,t)=
 \int_\R \frac{\partial}{\partial x}\left( \frac 1p
u(x-y,t)^{1/p-1}u_x(x-y,t)\right)v(y,t)^{1/q} \, dy =
 \]
 \[
\frac 1p \int_\R u(x-y,t)^{1/p-1}u_{xx}(x-y,t)v(y,t)^{1/q} \, dy +
 \]
 \be\label{f1}
\frac 1p\left( \frac 1p -1 \right)\int_\R
u(x-y,t)^{1/p-2}u_{x}^2(x-y,t)v(y,t)^{1/q} \, dy.
 \ee
Hence
 \[
\frac{1}p \int_\R u(x-y,t)^{1/p-1}u_{xx}(x-y,t )v(y,t)^{1/q} \, dy =
\]
 \be\label{f2}
\frac{\partial^2}{\partial x^2} h(x,t) + \frac 1{pp'}\int_\R
u(x-y,t)^{1/p}v(y,t)^{1/q}\left(\frac{u_{x}}{u}\right)^2(x-y,t) \, dy.
 \ee
Analogous formula for the last integral in \fer{f1}. Therefore we have
 \[
\frac{\partial}{\partial t}h(x,t) =  (\a + \b) \frac{\partial^2}{\partial
x^2} h(x,t) +
 \]
 \[
\frac{\a}{pp'}\int_\R
u(x-y,t)^{1/p}v(y,t)^{1/q}\left(\frac{u_{x}}{u}\right)^2(x-y,t) \, dy +
 \]
  \be\label{f3}
\frac{\b}{qq'}\int_\R
u(x-y,t)^{1/p}v(y,t)^{1/q}\left(\frac{v_{x}}{v}\right)^2(y,t) \, dy.
 \ee
Making use of formula \fer{f3}, we obtain
 \[
\frac d{dt }\int_\R h^r(x,t) \, dx =
 r \int_\R h(x,t)^{r-1} h_t(x,t) \, dx =
 r (\a +\b) \int_\R h^{r-1}(x,t) h_{xx}(x,t) \, dx +
 \]
 \[
\int_\R h^{r-1}(x,t) \left[ \frac{\a}{pp'}\int_\R
u(x-y,t)^{1/p}v(y,t)^{1/q}\left(\frac{u_{x}}{u}\right)^2(x-y,t) \, dy
+\right.
\]
\[
\left. \frac{\b}{qq'}\int_\R
u(x-y,t)^{1/p}v(y,t)^{1/q}\left(\frac{v_{x}}{v}\right)^2(y,t) \, dy
\right] \, dx.
\]
Since it holds
 \be\label{ff2}
\frac{\left( u^{1/p}\right)_x }{u^{1/p}} = \frac 1p \frac{ u_x }{u}, \quad
\frac{\left( v^{1/q}\right)_x }{v^{1/p}} = \frac 1q \frac{ v_x }{v}
 \ee
we obtain
 \[
\frac{\a}{pp'}\int_\R
u(x-y,t)^{1/p}v(y,t)^{1/q}\left(\frac{u_{x}}{u}\right)^2(x-y,t) \, dy = \a
\frac{p}{p'}\int_\R \frac{\left( u_x^{1/p}(x-y,t)\right)^2}{u^{1/p}(x-y)}
v(y,t)^{1/q}\, dy,
 \]
and
 \[
\frac{\b}{qq'}\int_\R
u(x-y,t)^{1/p}v(y,t)^{1/q}\left(\frac{v_{x}}{v}\right)^2(y,t)\, dy = \b
\frac{q}{q'}\int_\R u^{1/p}(x-y)\frac{\left(
v_y^{1/q}(y,t)\right)^2}{v^{1/q}(y)}\, dy.
 \]

 Finally
 \[
\frac 1r \frac d{dt}\int_\R h^r(x,t) \, dx =
  -(\a + \b) (r-1) \int_\R h^{r-2}(x,t) \left(h_x(x,t)\right)^2 \,
 dx +
 \]
 \be\label{acca1}
\a \frac{p}{p'} \int_\R h^{r-1}(x,t) A(u^{1/p},v^{1/q})(x,t) \, dx  + \b
\frac{q}{q'} \int_\R h^{r-1}(x,t) B(u^{1/p},v^{1/q})(x,t) \, dx.
 \ee
In \fer{acca1} we defined
 \be\label{a1}
A(f,g)(x,t) = \int_\R \frac{\left( f_x(x-y,t)\right)^2}{f(x-y)} g(y,t)\,
dy,
 \ee
and
 \be\label{b1}
B(f,g)(x,t) = \int_\R f(x-y)\frac{\left( g_y(y,t)\right)^2}{g(y)}\, dy.
 \ee
Since
 \[
\frac{d\Psi_{u,v}(t)}{dt} = \Psi_{u,v}(t)^{1-r} \, \frac d{dt}\int_\R
h^r(x,t) \, dx,
 \]
the sign of the time derivative of the functional $\Psi_{u,v}(t)$ depends
of the sign of the expression on the right-hand side of \fer{acca1}. In
order to determine this sign, the following Lemma will be of paramount
importance.

\begin{lem}\label{bl} Let $f(x)$ and $g(x)$  be probability density
functions such that both $A(f,g)$ and $B(f,g)$, given by \fer{a1} and
\fer{b1}, are well defined. Then, for all positive constants $a,b$ and $r
> 0$
 \[
\left( a^2 + b^2 + 2abr \right) \int_\R (f*g)^{r-2}\left( (f*g)_x\right)^2
\, dx \le
 \]
 \be\label{fish}
a^2 \int_\R (f*g)^{r-1} A(f,g)\, dx + b^2 \int_\R (f*g)^{r-1} B(f,g)\, dx.
 \ee
Moreover, there is equality in \fer{fish} if and only if, for any positive
constant $c$ and constants $m_1, m_2$,  $f$ and $g$ are Gaussian
densities, $f(x)= M_{ca}(x-m_1)$ and $g(x) = M_{cb}(x-m_2)$.
\end{lem}

\begin{proof}

The proof will follow along the same lines of the analogous one for Fisher
information, given by Blachman \cite{Bla}.  First of all, to easily
justify computations, let us prove the lemma by considering smooth
functions $f$ and $g$. Then the proof for general $f$ and $g$ will follow
owing to the convexity properties of $A$ and $B$ \cite{LT}. This can be
easily done by considering $f*M_t$ and $g*M_t$ solutions to the heat
equation for some $t>0$.  Let
\[
k(x) = f*g(x).
\]
Then, for any pair of  positive constants  $a,b$
 \[
(a+b)k'(x) = a \int_\R f'(x-y)g(y) \, dy + b \int_\R f(x-y)g'(y) \, dy.
 \]
Therefore
 \[
(a+b)\frac{k'(x)}{k(x)}  = a \int_\R \frac{f'(x-y)}{f(x-y)}
\frac{f(x-y)g(y)}{k(x)} \, dy + b \int_\R
\frac{g'(y)}{g(y)}\frac{f(x-y)g(y)}{k(x)} \, dy =
 \]
 \[
\int_\R \left(a \frac{f'(x-y)}{f(x-y)} + b \frac{g'(y)}{g(y)}\right) \,
d\mu_x(y),
 \]
where we denoted
 \[
d\mu_x(y) = \frac{f(x-y)g(y)}{k(x)} \, dy.
 \]
Note that, for every $x \in \R$, $d\mu_x$ is a unit measure on $\R$.
Consequently, by Jensen's inequality
 \[
(a+b)^2 \left[\frac{k'(x)}{k(x)}\right]^2 = \left[ \int_\R \left(a
\frac{f'(x-y)}{f(x-y)} + b \frac{g'(y)}{g(y)}\right) \, d\mu_x(y)
\right]^2 \le
 \]
 \be\label{jen}
\int_\R \left(a \frac{f'(x-y)}{f(x-y)} + b \frac{g'(y)}{g(y)}\right)^2 \,
d\mu_x(y).
 \ee
Hence, for every constant $r >0$
 \[
(a+b)^2 \int_\R k^r(x)\, \left[\frac{k'(x)}{k(x)}\right]^2 \, dx \le
 \]
 \[
\int_\R k^r(x)\, \int_\R \left(a \frac{f'(x-y)}{f(x-y)} + b
\frac{g'(y)}{g(y)}\right)^2 \,\frac{f(x-y)g(y)}{k(x)} \, dy  \, dx =
 \]
 \[
\int_\R k^{r-1} (x)\, \left[ a^2 \int_\R \frac{(f'(x-y))^2}{f(x-y)}g(y)\,
dy + b^2 \int_\R \frac{(g'(y))^2}{g(y)}f(x-y) \, dy \right] \, dx +
 \]
 \[
2ab \int_\R k^{r-1}(x) \int_\R f'(x-y)g'(y)\, dy \, dx.
 \]
On the other hand,
 \[
\int_\R f'(x-y)g'(y)\, dy  =  k''(x),
 \]
 so that
 \[
\int_\R k^{r-1}(x) \int_\R f'(x-y)g'(y)\, dy \, dx =
 \]
 \[
\int_\R k^{r-1}(x) r''(x) \, dx = - (r-1) \int_\R k^{r-2}(x)(k'(x))^2\,
dx.
 \]
This concludes the proof of the lemma. The cases of equality are easily
found from the following argument. Equality follows if, after application
of Jensen's inequality,  there is equality in \fer{jen}. On the other
hand, for any convex function $\varphi$ and unit measure $d\mu$ on the set
$\Omega$, equality in Jensen's inequality
 \[
 \varphi(\int_\Omega f\,d\mu) \le \int_\Omega \varphi(f) \, d\mu
 \]
holds true if and only if $f$ is constant, so that
 \[
f = \int_\Omega f\,d\mu.
 \]
In our case, this means that there is equality if and only if the
function
 \[
a \frac{f'(x-y)}{f(x-y)} + b \frac{g'(y)}{g(y)}
 \]
does not depend on $y$. If this is the case, taking the derivative
with respect to $y$, and using the identity
 \[
\frac d{dy}\left(\frac{f'(x-y)}{f(x-y)}\right) = -\frac
d{dx}\left(\frac{f'(x-y)}{f(x-y)}\right),
 \]
we conclude that $f$ and $g$ have to satisfy
 \be\label{eq3}
a \frac{d^2}{dx^2}\log f(x-y) = b \frac{d^2}{dy^2}\log g(y).
 \ee
Note that \fer{eq3} can be verified if and only if the functions on both
sides are constant. Thus, there is equality if and only if
 \be\label{m2}
\log f(x) = b_1x^2 + c_1 x + d_1, \quad \log g(x) = b_2x^2 + c_2 x +
d_2 .
 \ee
By coupling \fer{m2} with \fer{eq3}, we obtain that there is equality in
\fer{fish} if and only if $f$ and $g$ are gaussian densities, of variances
$ca$ and $cb$, respectively, for any given positive constant $c$.

\end{proof}

The case $r =1$ has been treated in Blachman \cite{Bla}, as part of
his proof of the entropy power inequality \fer{entr}. In this case
 \be\label{fisher}
I(f) = \int_\R \frac{(f'(x))^2}{f(x)} \, dx
 \ee
denotes the Fisher information of the probability density $f$, and
inequality \fer{fish} becomes
 \[
(a+b)^2 I(f*g) \le a^2 I(f) + b^2 I(g).
 \]
We remark that the validity of \fer{fish} is not restricted to probability
density functions. Indeed, it continues to hold for nonnegative functions
of any given mass.

Let us apply the result of Lemma \ref{bl} to control the sign of the
right-hand side in formula \fer{acca1}. If we choose $a^2 = \a
p/{p'}$, $b^2 = \b q/{q'}$ in \fer{fish},  then the coefficient of
the term on the left-hand side of inequality \fer{fish} assumes the
value
 \[
a^2 + b^2 + 2ab r = \a \frac p{p'} +  \b \frac q{q'} + 2\sqrt{ \a \b}
\sqrt{\frac{pq}{p'q'}} r.
 \]
Let us introduce, for any given $r >1$ the function
 \be\label{conv}
 \Gamma (\a, \b) = (\a + \b)(r-1) - \left(\a \frac p{p'} +  \b \frac q{q'} + 2\sqrt{ \a
 \b}\sqrt{\frac{pq}{p'q'}} r \right).
 \ee
It is clear that, as soon as for some values of $\a,\b$ the function
$\Gamma (\a,\b) \le 0$, the expression on the right-hand side of
\fer{acca1} is nonnegative, and the functional $\Psi_{u,v}(t)$ is
increasing.  In order to check its sign, consider that the function
$\Gamma$ is jointly convex, and it is such that, for any positive constant
$c$
 \[
\Gamma( c\a, c\b) = c \Gamma(\a,\b).
 \]
Therefore, if a point $(\a = \a_0, \b = \b_0)$ is an extremal point, also
the point $(c\a_0, c \b_0)$ is an extremal point, and $\Gamma$ admits the
half-line  $\b_0 \a = \a_0 \b$ of extremals. Since
 \[
\frac{\partial \Gamma}{\partial \a} = r- 1 -\frac p{p'} -
\sqrt{\frac \a\b} \sqrt{\frac{pq}{p'q'}} r,
 \]
by adding and subtracting the quantity $pr/q'$ we obtain
 \[
\frac{\partial \Gamma}{\partial \a} = r + \frac p{q'} r - 1 -\frac
p{p'} - \sqrt{\frac \a\b} \sqrt{\frac{pq}{p'q'}}r  + \frac p{q'} r =
 \]
 \[
pr\left( \frac 1p- \frac 1{q'}\right) - p\left( \frac 1p + \frac
1{p'} \right) - \sqrt{\frac \a\b} \sqrt{\frac{pq}{p'q'}}r  + \frac
p{q'} r.
 \]
Since
 \be\label{id1}
\frac 1p- \frac 1{q'} = \frac 1q- \frac 1{p'}= \frac 1r,
 \ee
one obtains
 \[
\frac{\partial \Gamma}{\partial \a} = - \sqrt{\frac \a\b}
\sqrt{\frac{pq}{p'q'}}r  + \frac p{q'} r = 0
 \]
if the point $(\a, \b)$ belong to the half-line
 \be\label{line}
\b = \frac p{q'}\cdot \frac{p'}q \a.
 \ee
Same result is obtained if we impose the vanishing of the partial
derivative of $\Gamma$ with respect to $\b$. On the other hand, thanks to
identity \fer{id1}
 \[
\Gamma\left( \frac{q'}p, \frac{p'}q \right) = \left( \frac{q'}p +
\frac{p'}q \right) (r-1) - \frac{q'}{p'} - \frac{p'}{q'} -2r =
 \]
 \[
q'r\left( \frac 1p - \frac 1{q'}\right) + p'r\left( \frac 1q - \frac
1{p'}\right) - q' -p' = 0.
 \]
Hence, along the line \fer{line}, in view of lemma \ref{bl} the
functional $\Phi_{u,v}(t)$ is increasing with respect to $t$.
Proceeding as in Section \ref{hol}, namely by scaling $u(x,t)$ and
$v(x,t)$ as in \fer{FP}, we conclude that the functional will keep
its maximum value as time goes to infinity, and
 \be\label{max-young}
\lim_{t\to \infty}\Phi_{u,v}(t) = \left( \int_\R u(x)\, dx
\right)^{1/p} \left( \int_\R v(x)\, dx \right)^{1/q} C(p,q,r),
 \ee
where
 \be\label{con}
 C(p,q,r) = \left( \int_\R \left(
M_{q'/p}(x)^{1/p}*M_{p'/q}(x)^{1/q}\right)^r \, dx. \right)^{1/r}.
 \ee
Using that the convolution of Gaussian functions is a Gaussian
function, $M_\a * M_\b = M_{\a + \b}$, we compute
 \[
 \left( \int_\R \left(
M_\a^{1/p}*M_{\b}^{1/q}\right)^r \, dx. \right)^{1/r} = \left[
\frac{p\a}{\a^{1/p}} \frac{q\b}{\b^{1/q}} \frac{(p\a
+q\b)^{1/r}}{r^{1/r}(p\a +q\b)} \right]^{1/2}.
 \]
The choice $\a = q'/p$, $\b = p'/q$ gives
 \[
C(p,q,r) =\left( A_pA_q A_{r'}\right)^{1/2},
 \]
where $A_m$ is defined by \fer{c+}. This concludes the proof of the first
part of Theorem \ref{th-young}.

The case in which $1/p + 1/q = 1+1/r$, but $ 0< p,q,r <1$ can be treated
likewise. In this case the dual exponents $p',q',r'$ are negative, and
formula \fer{acca1} takes the form
 \[
\frac 1r \frac d{dt}\int_\R h^r(x,t) \, dx =
  -\left[ -(\a + \b) (1-r) \int_\R h^{r-2}(x,t) \left(h_x(x,t)\right)^2 \,
 dx + \right.
 \]
 \be\label{acca2}
\left. \a \frac{p}{|p'|} \int_\R h^{r-1}(x,t) A(u^{1/p},v^{1/q})(x,t) \,
dx + \b \frac{q}{|q'|} \int_\R h^{r-1}(x,t) B(u^{1/p},v^{1/q})(x,t) \, dx
\right]
 \ee
To control the sign of the quantity into square brackets, we introduce now
the function
 \be\label{conv1}
 \tilde\Gamma (\a, \b) = (\a + \b)(1-r) - \left(\a \frac p{|p'|} +  \b \frac q{|q'|} + 2\sqrt{ \a
 \b}\sqrt{\frac{pq}{|p'||q'|}} r \right).
 \ee
In this case
 \[
\frac{\partial \Gamma}{\partial \a} =  1 -r  +\frac p{p'} + \sqrt{\frac
\a\b} \sqrt{\frac{pq}{|p'||q'|}} r.
 \]
By adding and subtracting the quantity $pr/q'$ we obtain as before
 \[
\frac{\partial \Gamma}{\partial \a} = \sqrt{\frac \a\b}
\sqrt{\frac{pq}{|p'||q'|}}r  + \frac p{q'} r = 0
 \]
if the point $(\a, \b)$ belong to the half-line
 \be\label{line}
\b = \frac p{|q'|}\cdot \frac{|p'|}q \a.
 \ee
This choice however implies that the right-hand side in \fer{acca2} is non
positive, and the functional $\Psi_{u,v}(t)$ decreases. This leads to the
reverse Young's inequality \fer{young-r}.

\end{proof}

\section{Conclusions}

In this paper we presented a new proof of the sharp form of Young's
inequality for convolutions, as well as and its reverse form. For the sake
of simplicity, this proof has been done in dimension $n=1$. Looking at the
details of the computations, it appears evident that the proof still holds
in dimension $n >1$, since the computations in higher dimension do not
affect the constants both in formula \fer{acca1} and \fer{fish}, which are
at the basis of the whole procedure. The main difference relays in the
fact that the Gaussian functions are $n$-dimensional Gaussians, which lead
to the additional presence of the exponent $n$ in the sharp constant.
Hence Theorem \ref{th-young} leads to the sharp inequality \fer{young} in
any dimension, without any additional (if not computational) difficulty.
Also, both Young's inequality and its reverse form are here derived by a
unique well understandable physical principle, in the form of time
monotonicity of a Lyapunov functional.
\bigskip \noindent

{\bf Acknowledgment:} The author acknowledge support by MIUR project
``Optimal mass transportation, geometrical and functional
inequalities with applications''.

\end{document}